\documentclass[letterpaper, 10 pt, conference]{ieeeconf}
\IEEEoverridecommandlockouts
\usepackage[utf8]{inputenc}
\usepackage[T1]{fontenc}
\usepackage{amsmath,amssymb,amsfonts}

\usepackage{amsthm}
\usepackage{bm}
\usepackage{mathtools}
\usepackage{graphicx}
\usepackage{subcaption}
\let\labelindent\relax
\usepackage{enumitem}
\usepackage{cite}
\usepackage{algorithm}
\usepackage{algpseudocode}
\usepackage{hyperref}
\hypersetup{colorlinks=true,linkcolor=red,citecolor=black,urlcolor=black}
\usepackage{xcolor}

\newtheorem{theorem}{Theorem}
\newtheorem{proposition}{Proposition}

\theoremstyle{remark}
\newtheorem{remark}{Remark}
\theoremstyle{definition}
\newtheorem{definition}{Definition}

\newcommand{\R}{\mathbb{R}}

\newcommand{\ii}{\mathrm{i}}
\newcommand{\Supp}{\mathrm{h}} 
\newcommand{\Ell}{E}           
\newcommand{\CEll}{C}          

\title{
	Shaping Energy Exchange with Gyroscopic Interconnections: a Geometric Approach
}
\author{Jasper Juchem$^{1}$ and Mia Loccufier$^{1}$%
\thanks{*This work was not supported by any organization}%
\thanks{$^{1}$Jasper Juchem (\textit{corresponding author}) and Mia Loccufier are with the Dept. of Electromechanical, Systems and Metal Engineering, Ghent University
	{\tt\small \{Jasper.Juchem;Mia.Loccufier\}@UGent.be}}%
}

\begin{document}

	
	\maketitle
	\thispagestyle{empty}
	\pagestyle{empty}
	
	\begin{abstract}
	Gyroscopic interconnections enable redistribution of energy among degrees of freedom while preserving passivity and total energy, and they play a central role in controlled Lagrangian methods and IDA–PBC. Yet their quantitative effect on transient energy exchange and subsystem performance is not well characterised.
	
	We study a conservative mechanical system with constant skew‑symmetric velocity coupling. Its dynamics are integrable and evolve on invariant two‑tori, whose projections onto subsystem phase planes provide a geometric description of energy exchange. When the ratio of normal‑mode frequencies is rational, these projections become closed resonant Lissajous curves, enabling structured analysis of subsystem trajectories.
	
	To quantify subsystem behaviour, we introduce the inscribed‑radius metric: the radius of the largest origin‑centred circle contained in a projected trajectory. This gives a lower bound on attainable subsystem energy and acts as an internal performance measure. We derive resonance conditions and develop an efficient method to compute or certify the inscribed radius without time‑domain simulation.
	
	Our results show that low‑order resonances can strongly restrict energy depletion through phase‑locking, whereas high‑order resonances recover conservative bounds. These insights lead to an explicit interconnection‑shaping design framework for both energy absorption and containment control strategies, while taking responsiveness into account.
	\end{abstract}
	
	
	\section{Introduction}
	Energy-based control provides structure-preserving tools to shape closed-loop dynamics by manipulating storage, exchange, and dissipation while retaining passivity. Two complementary paradigms, controlled Lagrangians (CL) and interconnection-and-damping assignment passivity-based control (IDA--PBC), systematically modify kinetic/potential energies and the interconnection structure to achieve stabilisation and performance objectives without destroying the underlying mechanics \cite{BlochCL1,BlochCL2,Ortega2002,MaschkeVdS1992,vanDerSchaftL2, Woolsey2004}. In both, gyroscopic interconnections, which are skew-symmetric velocity couplings, play a central role: they redistribute energy among degrees of freedom while preserving the total energy and the (port-)Hamiltonian structure \cite{Blankenstein2002,Duindam2009,Willems2007,HillMoylan1976}. 
	
	Despite their widespread use, existing approaches primarily emphasise stability and structural properties, while the effect of gyroscopic interconnections on transient energy exchange and subsystem performance remains poorly quantified. In particular, there is a lack of design-ready metrics that predict how interconnection choices limit or enable energy transfer between subsystems and how resonance constrains attainable performance. Addressing this gap is essential for applications where energy routing, absorption, or containment is a primary objective.
	
	In this paper, we study a minimal conservative two-degree-of-freedom system with constant skew-symmetric velocity coupling, viewed as a canonical closed-loop normal form induced by interconnection shaping. The dynamics are integrable and evolve on invariant two-tori, and subsystem energy exchange can be interpreted geometrically through projections onto subsystem phase planes. This perspective reveals that the attainable transient behaviour is governed by the resonance structure of the modal frequencies, which is directly controlled by the interconnection parameter.
	
	From a control perspective, the interconnection strength provides a tunable design parameter that selects the resonance class and thereby the achievable depth and speed of energy exchange. This leads naturally to two complementary design objectives: absorption, where energy is transferred away from a subsystem, and containment, where a guaranteed energy level is preserved. The resulting trade-off between exchange depth and responsiveness can be interpreted as a Pareto frontier induced by the resonance structure.
	
	Related ideas of tunable coupling and mode interaction appear in engineered	resonators, where coupling is modulated to route energy	between modes. Applications include dynamic electrostatic modulation in gyroscopic ring resonators, frequency-mismatch compensation and mode matching in dual-mass gyroscopes, and analyses that explicitly target energy transfer between drive and sense modes \cite{Zhou2019,Pistorio2021,Hamed2016,CouplingMEMS2021}. These strands motivate a control-oriented quantification of energy routing induced by gyroscopic interconnections. Applications that benefit from targeted energy routing as well include vibration mitigation and energy funnelling toward single actuation points \cite{Juchem2024}.
	
	The contributions of this work are as follows: (i) a geometric characterisation of subsystem energy exchange via invariant-torus projections; (ii) exact resonance conditions and a degeneracy criterion for vanishing inscribed radius; (iii) a computationally efficient framework for certifying the minimal subsystem energy; and (iv) an interconnection-shaping design perspective that explicitly links resonance structure to transient performance. 
	

	\section{Model and Resonant pairs}\label{sec:modelResonantPairs}
	\subsection{Canonical Model}
	Consider
	\begin{equation}\label{eq:model}
		\ddot q + n\, \dot z + q = 0, \qquad \ddot z - n\, \dot q + z = 0,
	\end{equation}
	with state $(q,\dot q,z,\dot z)\in\R^4$ and constant gyroscopic coupling $n\in\R$. 
	System \eqref{eq:model} should be viewed as a canonical closed-loop normal form that arises when energy-shaping methods (controlled Lagrangians / IDA–PBC) introduce a constant skew-symmetric interconnection between two oscillatory coordinates. Define $x=(q,z)^\top$, $p=(\dot q,\dot z)^\top$, and the skew matrix $J=\begin{bmatrix}0&n\\-n&0\end{bmatrix}$. Then \eqref{eq:model} is equivalent to
	\begin{equation}
		\dot x = p,\qquad \dot p = -x - J p.
	\end{equation}
	With the standard symplectic form, the Hamiltonian $H(x,p)=\tfrac12(\|x\|^2+\|p\|^2)$ yields the dynamics above. Moreover, $\dot H = p^\top \dot p + x^\top \dot x = -p^\top J p = 0$ since $J^\top=-J$. Hence, the coupling is called `gyroscopic' in the sense that it is skew-symmetric and conserves total energy while redistributing it between subsystems \cite{ArnoldMMCM,vanDerSchaftL2}. The Hamiltonian of a subsystem is defined as
	\begin{equation}
		H_k(k, \dot{k}) = \frac{1}{2}(k^2 + \dot{k}^2)
	\end{equation} 
	with $k\in \{q, z\}$.
	
	In many multi-DOF mechanical plants, a dominant mode coupled to an auxiliary (virtual) oscillator via interconnection shaping reduces locally to \eqref{eq:model} after modal truncation and normalisation \cite{Juchem2024}. Here, the coupling $n$ is a tunable controller parameter rather than a physical constant. A concrete physical instantiation also arises in electromechanical resonators, where a single-DOF mass–spring mode is coupled to a coil closed through an ideal capacitor, such that the electromagnetic transduction defines a power-preserving port interconnection \cite{Auleley2021}. This interpretation emphasises that \eqref{eq:model} captures a minimal energy-preserving interconnection structure that appears both in physical systems and as a design primitive in control synthesis.

	\subsection{Resonant pairs}
	Introduce $u=q+\ii z$. Then \eqref{eq:model} reduces to
	\begin{equation}\label{eq:u}
		\ddot u - \ii n \dot u + u = 0.
	\end{equation}
	The characteristic equation for $u\sim e^{\lambda t}$ is $\lambda^2 - \ii n \lambda + 1 = 0$, with imaginary roots $\lambda=\ii\Omega_1$ and $\lambda=-\ii\Omega_2$ where we define the \emph{positive} modal frequencies
	\begin{equation}\label{eq:Omegas}
		\Omega_1 := \frac{\sqrt{n^2+4}+n}{2},\quad \Omega_2 := \frac{\sqrt{n^2+4}-n}{2},\quad \Omega_1\Omega_2=1.
	\end{equation}
	Thus the general solution is
	\begin{equation}\label{eq:generalSol}
		u(t)=C_1 e^{\ii\Omega_1 t}+C_2 e^{-\ii\Omega_2 t},
	\end{equation}
	a quasi-periodic flow on invariant two-tori in the full 4D phase space \cite{KatokHasselblatt}.
	
	\begin{definition}[Resonant pair]
		Let $\Omega_1,\Omega_2>0$ be the modal frequencies defined in \eqref{eq:Omegas}. 
		A pair $(\tau,\sigma)\in\mathbb{N}^2$ with $\gcd(\tau, \sigma)=1$, also referred to as being \emph{coprime}, is called a \emph{resonant pair} if
		\[
		\frac{\Omega_1}{\Omega_2}=\frac{\tau}{\sigma} \in\mathbb{Q}.
		\]
	\end{definition}
	Without loss of generality, we assume that $\tau > \sigma$.
	
	\begin{proposition}[Resonance condition]\label{prop:resonance}
		The projected motions on $(q,\dot q)$ and $(z,\dot z)$ are periodic (closed Lissajous curves) if and only if $(\Omega_1, \Omega_2)$ is a resonant pair. In terms of $n$ this is equivalent to
		\begin{equation}\label{eq:n_res}
			n^2=\frac{(\tau-\sigma)^2}{\tau\sigma}.
		\end{equation}
	\end{proposition}
	\begin{proof}
		The ratio of basic frequencies is rational if and only if the 2-torus flow is periodic when projected to time \cite[Proposition 1.5.1, p. 33]{KatokHasselblatt}. Solving $\frac{\Omega_1}{\Omega_2}=\frac{\sqrt{n^2+4}+n}{\sqrt{n^2+4}-n}=\frac{\tau}{\sigma}$ for $n$ yields \eqref{eq:n_res}.
	\end{proof}
	We discern two classes of resonant pairs:
	\begin{definition}[Low-order resonance pair]
 		A resonant pair $(\tau, \sigma)$ is called \emph{low-order} if it belongs to a fixed finite set
 		\[
 		\mathcal{R}_{\mathrm{low}}
 		=\bigl\{(\tau,\sigma)\in\mathbb{N}^2:\gcd(\tau,\sigma)=1,\; \tau+\sigma\le M\bigr\},
 		\]
 		for some $M\in\mathbb{N}$.
	 	Low-order resonance pairs correspond to strong phase locking between the modal components, a phenomenon represented by the characteristically wide Arnold tongues associated with small ratios \cite[Section 3.2.4]{Pikovsky}.
	 \end{definition}
	 
	 \begin{definition}[High-order resonance pair]
	 	A sequence of resonant pairs $\{(\tau_k,\sigma_k)\}$ is called \emph{high-order} if
	 	\[
	 	\tau_k + \sigma_k \to\infty\quad\text{and}\quad
	 	\frac{\tau_k}{\sigma_k}\;\to\;1,
	 	\text{ with}\quad
	 	(\tau_k,\sigma_k)\neq(\tau,\tau),
	 	\]
	 	equivalently if
	 	\[
	 	\tau_k,\sigma_k\to\infty
	 	\quad\text{and}\quad
	 	|n|=\frac{|\tau_k-\sigma_k|}{\sqrt{\tau_k \sigma_k}}\to 0.
	 	\]
	 	Hence, high-order resonance pairs correspond to weak phase locking and asymptotically dense sampling of the invariant torus \cite{Pikovsky}.
	 \end{definition}

	\section{Subsystem Projection: Lissajous curves}\label{sec:subsystem}
	For the impulse disturbance class
	\begin{equation}\label{eq:IC}
		\mathcal{D} := \{(q,\dot q,z,\dot z)(0)=(0,\,\dot{q}_0,\,0,\,0): |\dot{q}_0| \leq D\},
	\end{equation}
	with $D\in\mathbb{R}_{>0}$,
	solving \eqref{eq:generalSol} for $C_1,C_2$ gives $C_2=-C_1=\ii\rho_0$ with
	\begin{equation}\label{eq:rho0}
		\rho_0:=\frac{\dot{q}_0}{\Omega_1+\Omega_2}=\frac{\dot{q}_0}{\sqrt{n^2+4}}.
	\end{equation}
	Hence
	\begin{equation}\label{eq:impulseResp-explicit} 
	\begin{cases}
		q(t)&=\rho_0\big(\sin(\Omega_1 t) + \sin(\Omega_2 t)\big),\\
		z(t)&=\rho_0\big(\cos(\Omega_2 t) - \cos(\Omega_1 t)\big),\\
		\dot q(t)&=\rho_0\big(\Omega_1\cos(\Omega_1 t)+\Omega_2\cos(\Omega_2 t)\big),\\
		\dot z(t)&=\rho_0\big(\Omega_1\sin(\Omega_1 t)-\Omega_2\sin(\Omega_2 t)\big).
	\end{cases}
	\end{equation}
	
	We focus now on the projection on the $(q,\dot{q})$ phase space. Similar results can be obtained for the projection on the $(z, \dot{z})$ phase space, due to structure of system \eqref{eq:model}. 	
	\begin{proposition}[Convex envelope on $(q,\dot q)$]\label{prop:convexEnvelope}
		Let $(q,\dot q)$ evolve under \eqref{eq:model} with the impulse disturbance \eqref{eq:IC}, then the boundary set is given by
		\begin{equation}\label{eq:boundary}
			\frac{x_{\partial}(\phi)}{\rho_0}
			=
			\frac{(\cos\phi,\;\Omega_1^2\sin\phi)}{\sqrt{\cos^2\phi+\Omega_1^2\sin^2\phi}}
			+\frac{(\cos\phi,\;\Omega_2^2\sin\phi)}{\sqrt{\cos^2\phi+\Omega_2^2\sin^2\phi}},
		\end{equation}with $\phi\in[0,2\pi]$.
	\end{proposition}
	
	\begin{proof}
		Let $\theta_i=\Omega_i t$, and define the ellipse boundary and filled ellipse respectively
		\begin{align}
			\CEll_\omega &:=\{(\sin\theta,\;\omega\cos\theta):\theta\in[0,2\pi]\},\\
			\Ell_\omega &:=\{(x,y):\,x^2+(y/\omega)^2\le 1\}.
		\end{align}
		From \eqref{eq:impulseResp-explicit}, write
		\begin{equation}
			(q,\dot q)(t)
			=\rho_0\bigl(\sin\theta_1,\Omega_1\cos\theta_1\bigr)
			+\rho_0\bigl(\sin\theta_2,\Omega_2\cos\theta_2\bigr),
		\end{equation}
		which yields the exact projected set $\mathcal{S}
		  	=\rho_0\,(\CEll_{\Omega_1}+\CEll_{\Omega_2}).$
		Taking convex hulls $\mathrm{co}(A+B)=\mathrm{co}(A)\oplus\mathrm{co}(B)$ together with $\mathrm{co}(C_\omega)=E_\omega$, where $E_\omega=\{(x,y):x^2+(y/\omega)^2\leq1\}$, gives
		$\operatorname{co}(\mathcal{S})=\rho_0\,(\Ell_{\Omega_1}\oplus\Ell_{\Omega_2}),
		$ with $\mathrm{\oplus}$ the Minkowski addition \cite{Schneider2014}. 
		Using additivity of two ellipsis' support functions,
		\begin{equation}
			\Supp_{\operatorname{co}(\mathcal{S})}(u(\phi)) = \rho_0\left(\Supp_{\Ell_{\Omega_1}}(u(\phi)) + \Supp_{\Ell_{\Omega_2}}(u(\phi))\right),
		\end{equation}
		which yields the parametric boundary \eqref{eq:boundary}.
		

	\end{proof}
	
	\begin{remark}[Ellipticity]\label{rem:n0ellipse}
		The curve \eqref{eq:boundary} is an ellipse \emph{iff} $\Omega_1=\Omega_2=1$, i.e.\ $n=0$.  
		Moreover, for $n\neq 0$ the boundary is the Minkowski addition of two unequal ellipses and is never itself an ellipse.  
	\end{remark}
	
	A similar analysis can be done for $(z, \dot{z})$. In Figure \ref{fig:boundary}, the boundary set and the maximal radius is plotted for $n=2$ and $n={1}/{\sqrt{12}}$, which is the resonant pair $(\tau, \sigma)=(4, 3)$ according to \eqref{eq:n_res}. Notice that in Figure \ref{fig:reachable_n2} for $n=2$, in absence of a resonant pair, the solution set is dense, and the envelope is clearly not elliptical. Figure \ref{fig:reachable_n1_sqrt12} shows a closed Lissajous curve, thus the set is sparse.
	\begin{figure}[t!]
		\centering
		\begin{subfigure}{0.45\columnwidth}
			\centering
			\includegraphics[width=\textwidth]{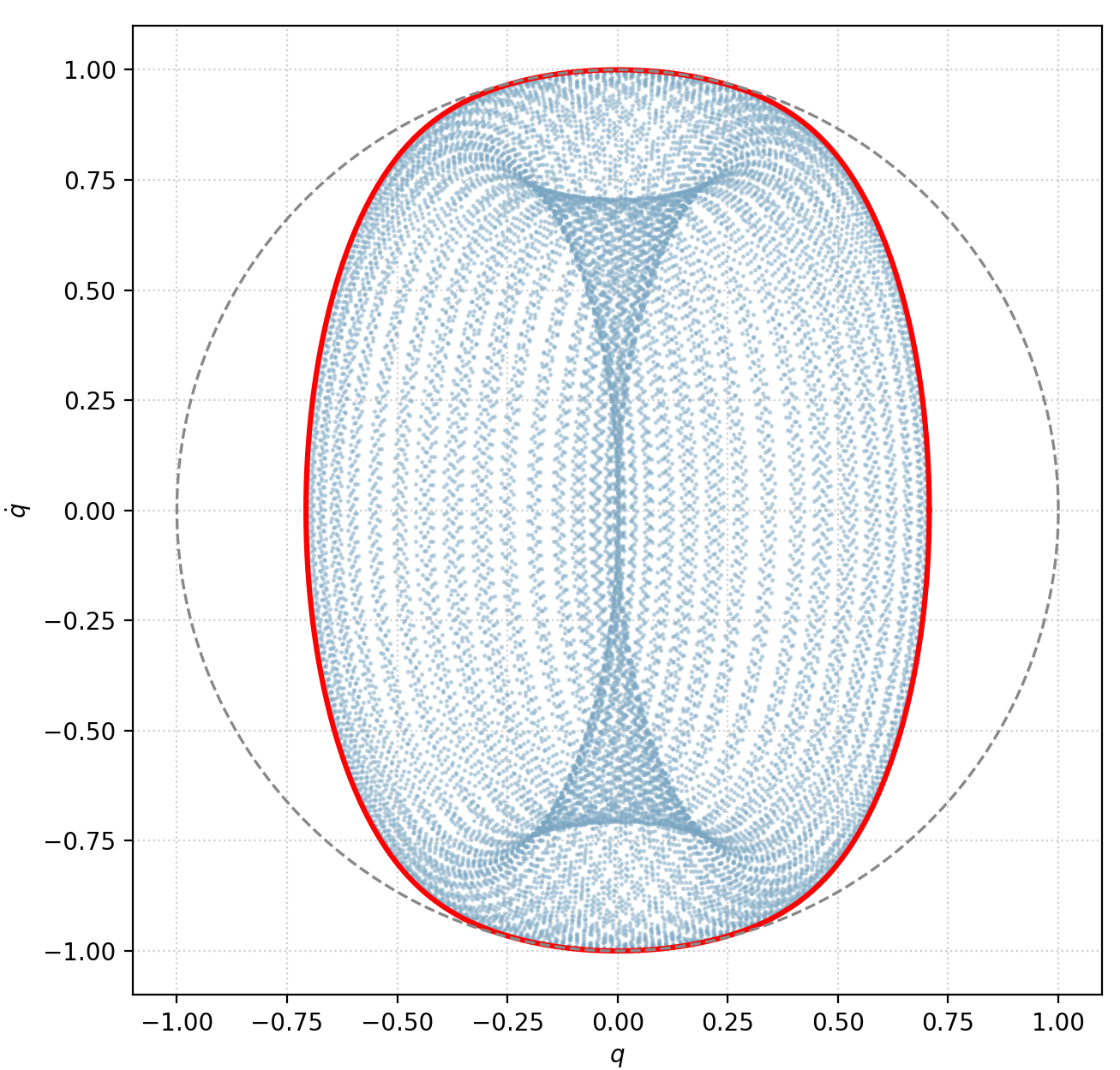}
			\caption{}
			\label{fig:reachable_n2}
		\end{subfigure}
		\begin{subfigure}{0.45\columnwidth}
			\centering
			\includegraphics[width=\textwidth]{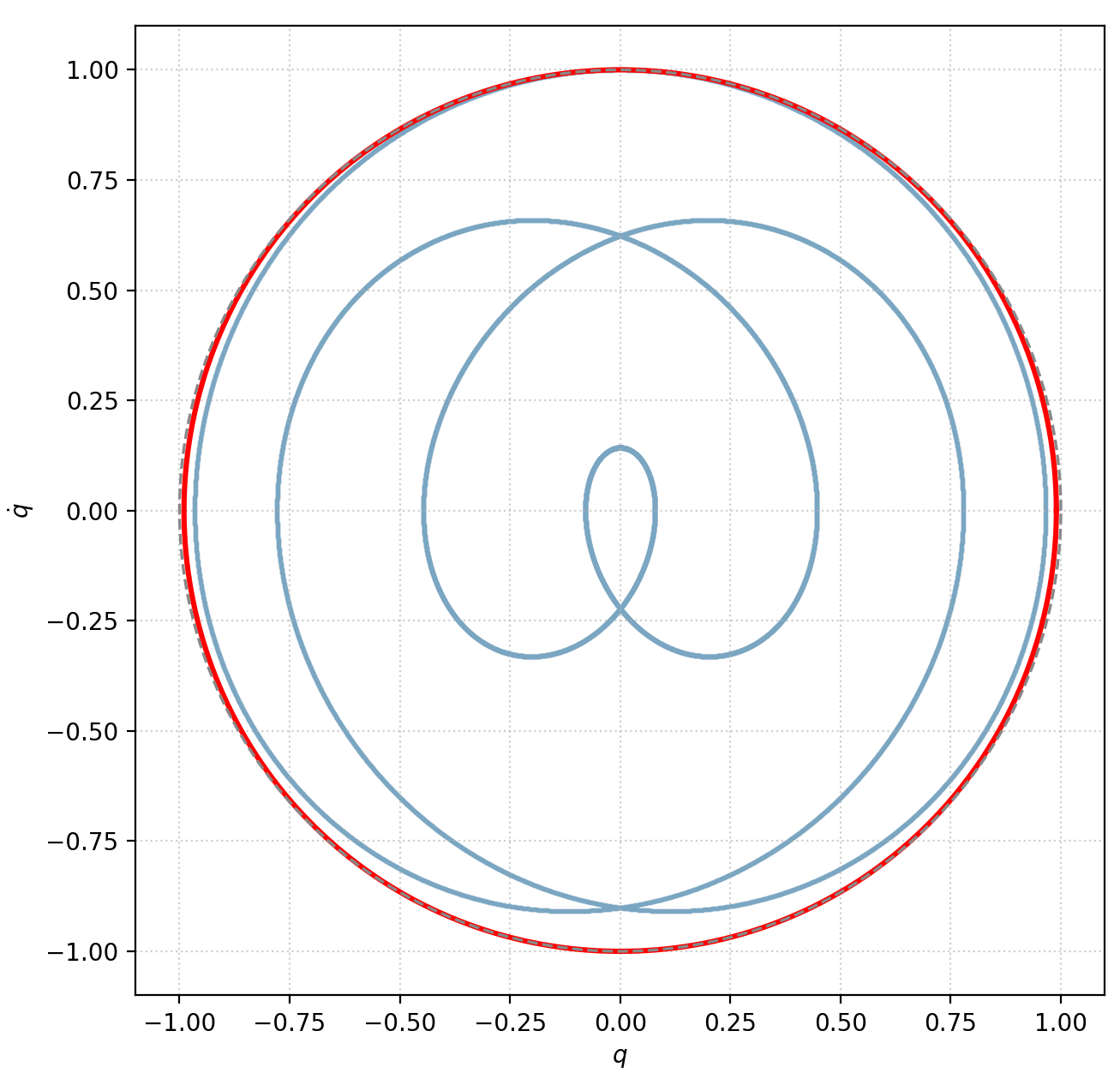}
			\caption{}
			\label{fig:reachable_n1_sqrt12}
		\end{subfigure}
		
		\caption{The phase plane $(q(t), \dot{q}(t))$, $t\in[0, \frac{200\pi}{n}]$, of \eqref{eq:model} for an impulse response ($\dot{q}_0=1$) with (a) $n=2$ and (b) $n=\frac{1}{\sqrt{12}}$. The boundary set $x_\partial$, given by \eqref{eq:boundary} (full line, red), and maximal radius (dotted line) are plotted as well.}
		\label{fig:boundary}
	\end{figure}
	
	\subsection{Resonant Lissajous curves}
	At a resonant pair $(\tau, \sigma)$, define $\theta:=\alpha t$, so $q,\dot q$ depend on $\theta$ only through harmonics $\tau\theta$ and $\sigma\theta$. From $\Omega_1\Omega_2=1$ follows that $\Omega_1=\Omega_2^{-1}=\sqrt{\frac{\tau}{\sigma}}$. So $(\tau\alpha, \sigma\alpha) = (\Omega_1, \Omega_2)$ iff $\alpha=\frac{1}{\sqrt{\tau\sigma}}$. Consider now the impulse response \eqref{eq:impulseResp-explicit}:
	\begin{equation}\label{eq:q_theta}
		\begin{cases}
		q(\theta)=\rho_0\left(\sin(\tau\theta) + \sin(\sigma\theta)\right),\\
		\dot{q}(\theta)=\rho_0\left(\sqrt{\frac{\tau}{\sigma}}\cos(\tau\theta) + \sqrt{\frac{\sigma}{\tau}}\cos(\sigma\theta)\right),
		\end{cases}
	\end{equation}
	and the squared radius
	\begin{equation}\label{eq:R_init}
		R(\theta)= r(\theta)^2 = q(\theta)^2+\dot q(\theta)^2.
	\end{equation}
	
	\subsection{Resonant inscribed circle}
	Define the resonant inscribed radius
	\begin{equation}
		r_{res} := \min_{\theta\in[0,2\pi[} \sqrt{R(\theta)}.
	\end{equation}
	 The inscribed circle tells us the minimal amount of energy $H_{q, min} = \frac{1}{2}r_{res}^2$ that remains in the subsystem $(q, \dot{q})$ at all time. This is an interesting metric to evaluate or design a system if the coupling parameter  $n$ of \eqref{eq:model} is tunable. The goal of this section is to understand the inscribed circle in the case of a resonant pair. Note that in non-resonant cases, the solution is dense in the phase space, as shown in Figure \ref{fig:reachable_n2}, so it tends to have a very small inscribed circle.
	
	\begin{theorem}[Degenerate resonant inscribed circle]
		Let $(\tau, \sigma)$ be a resonant pair and consider the resonant parametrisation \eqref{eq:q_theta}. Then
		\[
		r_{res}=0 \quad\Leftrightarrow\quad \delta:=\tau - \sigma\equiv 2 \pmod{4}.
		\]
		Moreover, in that case $R(\pi/2)=0$, so the minimum is attained at
		$\theta^\star=\pi/2 \ (\mathrm{mod}\ \pi)$.
	\end{theorem}
	
	\begin{proof}
		Using \eqref{eq:q_theta} and \eqref{eq:R_init}
		\begin{multline*}
		\frac{R(\theta)}{\rho_0^2}
		=\big(\sin(\tau\theta)+\sin(\sigma\theta)\big)^2
		\\
		+ \left(\sqrt{\frac{\tau}{\sigma}}\cos(\tau\theta)+		\sqrt{\frac{\sigma}{\tau}}\cos(\sigma\theta)\right)^2,
		\end{multline*}
		so $R(\theta)=0$ iff
		\begin{equation}\label{eq:R_0}
			\begin{cases}
			\sin(\tau\theta)+\sin(\sigma\theta)=0,\\
			\sqrt{\frac{\tau}{\sigma}}\cos(\tau\theta)+\sqrt{\frac{\sigma}{\tau}}\cos(\sigma\theta)=0.
			\end{cases}
		\end{equation}
		Using the product-to-sum formula,
		\[
		\sin(\tau\theta)+\sin(\sigma\theta)
		=2\sin\!\Big(\tfrac{\tau+\sigma}{2}\theta\Big)\cos\!\Big(\tfrac{\tau-\sigma}{2}\theta\Big),
		\]
		the first equation in \eqref{eq:R_0} holds if either (a) $\cos\big(\tfrac{\delta}{2}\theta\big)=0$ or (b) $\sin\big(\tfrac{\tau+\sigma}{2}\theta\big)=0$.
		\begin{itemize}		
		\item\emph{Case (a).} If $\cos\!\big(\tfrac{\delta}{2}\theta\big)=0$, then
		$
		\tau\theta=\sigma\theta+\pi \quad (\mathrm{mod}\ 2\pi),
		$
		so $\cos(\tau\theta)=-\cos(\sigma\theta)$. Substituting into the second
		equation in \eqref{eq:R_0} gives
		$
		\cos(\sigma\theta)
		\left(\sqrt{\tfrac{\sigma}{\tau}}-\sqrt{\tfrac{\tau}{\sigma}}\right)=0.
		$
		Since $\tau\neq\sigma$, the coefficient is nonzero and therefore
		$\cos(\sigma\theta)=0$. Hence
		$
		\theta=\frac{\pi}{2\sigma}(1+2k),\ k\in\mathbb{Z}.
		$
		On the other hand, from $\cos\!\big(\tfrac{\delta}{2}\theta\big)=0$ we obtain
		$
		\theta=\frac{\pi}{\delta}(1+2m),\ m\in\mathbb{Z}.
		$
		Equating these two expressions for $\theta$ yields the Diophantine condition
		$
		\delta(1+2k)=2\sigma(1+2m).
		$
		Taking $2$-adic valuations, the exponent of $2$ in the prime factorisation, on both sides gives
		$
		v_2(\delta)=1+v_2(\sigma).
		$
		Because $\gcd(\tau,\sigma)=1$, at most one of
		$\tau$ or $\sigma$ can be even. If $\sigma$ were even, then $\tau$ is odd and
		$\delta$ would be odd, so $v_2(\sigma)=-1$, contradicting $\sigma\in\mathbb{N}$. Hence
		$\sigma$ must be odd, so $v_2(\sigma)=0$ and therefore $v_2(\delta)=1$. Thus
		$\delta=2(1+2r), \ r\in\mathbb{Z}$, i.e.
		$
		\delta\equiv 2 \pmod{4}.
		$		
		\item\emph{Case (b).} If $\sin\!\big(\tfrac{\tau+\sigma}{2}\theta\big)=0$, then
		$
		\tfrac{\tau+\sigma}{2}\theta=\pi k,\qquad k\in\mathbb{Z},
		$
		and therefore $\cos(\tau\theta)=\cos(\sigma\theta)$. Substituting this into the
		second equation of \eqref{eq:R_0} yields
		$
		\cos^2(\sigma\theta)(\tau-\sigma)=0.
		$
		Since $\tau\ne\sigma$, the coefficient is nonzero, hence $\cos(\sigma\theta)=0$.
		Thus
		$
		\theta=\frac{\pi}{2\sigma}(1+2m),\ m\in\mathbb{Z}.
		$
		Together with $\tfrac{\tau+\sigma}{2}\theta=\pi k$, this gives
		$
		(\tau+\sigma)(1+2m)=4\sigma k.
		$
		Taking $2$-adic valuations on both sides and using $v_2(1+2m)=0$ gives
		$
		v_2(\tau+\sigma)=2+v_2(\sigma)+v_2(k).
		$
		
		Since $\gcd(\tau,\sigma)=1$, at most one of them is even. If $\sigma$ were even,
		then $\tau$ would be odd and $\tau+\sigma$ odd, so $v_2(\tau+\sigma)=0$,
		contradicting the above equality. Hence $\sigma$ is odd, so $v_2(\sigma)=0$ and
		therefore $v_2(\tau+\sigma)\ge2$. In particular, both $\tau$ and $\sigma$ are
		odd. Using
		\begin{equation*}
			\begin{cases}
				2\tau=(\tau+\sigma)+(\tau-\sigma),\\
				2\sigma=(\tau+\sigma)-(\tau-\sigma),
			\end{cases}
		\end{equation*}
		and the fact that $v_2(x+y)=\min(v_2(x),v_2(y))$ when $v_2(x)\ne v_2(y)$, we
		obtain from $v_2(2\tau)=v_2(2\sigma)=1$ and $v_2(\tau+\sigma)\ge2$ that
		$v_2(\tau-\sigma)=1$. Hence
		$
		\delta\equiv 2 \pmod{4}.
		$ Hence, both branches lead to the same criterion.
		\end{itemize}
		
		Finally, if $\delta\equiv 2\ (\mathrm{mod}\ 4)$ then $\tau,\sigma$ are odd with opposite residues mod $4$,
		and at $\theta=\pi/2$ one has $\cos(\tau\pi/2)=\cos(\sigma\pi/2)=0$ while $\sin(\tau\pi/2)=-\sin(\sigma\pi/2)$, giving $R(\pi/2)=0$.
	\end{proof}
	The classification of resonant pairs $(\tau,\sigma)$ for which the inscribed radius
	$r_{\mathrm{res}}$ vanishes plays a central conceptual role in our analysis.
	Indeed, the condition $\delta\equiv 2\!\!\pmod 4$ characterises precisely those
	commensurate frequency ratios for which the two modal ellipses become exact
	negatives of each other at a common phase, thereby producing complete
	cancellation of both displacement and velocity.  This phenomenon is highly
	non-generic: it does not depend on the amplitude $\rho_0$ nor on the individual
	modal frequencies, but solely on an arithmetic property of the resonance
	pair. For the control interpretation, these resonances
	mark precisely the regimes where passive gyroscopic coupling permits maximal
	energy exchange between modes. Consequently, the congruence condition
	$\delta\equiv 2\pmod 4$ is not merely a number-theoretic curiosity but a
	structurally meaningful criterion governing the attainable performance limits
	of the system.
	
	Now, consider the case where $\delta\neq 2 \pmod{4}$.
	\begin{theorem}\label{thm:globalMinimiser}
		Consider the parametric Lissajous curve \eqref{eq:q_theta} with squared radius \eqref{eq:R_init}, and assume $\delta\not\equiv 2 \pmod 4$.  
		Let $\{\theta_i\}$ be the ordered roots of $q(\theta)=0$ and $I_i=[\theta_i,\theta_{i+1}]$ the corresponding lobes of the Lissajous curve.  
		Let $\theta_{c,i}\in]\theta_i,\theta_{i+1}[$ be the unique root of $\dot{q}=0$ in lobe $I_i$.  
		Let 
		\[
		\alpha(\theta)=\frac{\tau + \sigma}{2}\theta,\qquad \beta(\theta)=\frac{\tau-\sigma}{2}\theta.
		\]
		A lobe $I_i$ will be called {\it envelope-minimising} if 
		\[
		\min_{\theta\in I_i}|\cos\beta(\theta)|
		\]
		is globally minimal among all lobes. Then:
		
		\begin{enumerate}
			\item The global minimum of $r(\theta)$ is attained in an envelope-minimising lobe.
			
			\item In any envelope-minimising lobe $I_i$, the function $R(\theta)$ is strictly decreasing on $[\theta_i,\theta_{c,i}]$ and strictly increasing on $[\theta_{c,i},\theta_{i+1}]$.  
			Consequently, $R(\theta)$ admits a unique minimiser in $I_i$, namely $\theta_{c,i}$.
			
			\item Therefore
			\[
			\min_{\theta\in[0,2\pi]} r(\theta)
			=\min_i r(\theta_{c,i})
			=\min_i |q(\theta_{c,i})|.
			\]
		\end{enumerate}
		\label{thm:main}
	\end{theorem}
	
	\begin{proof}
		Introduce the constants
		\[
		A=\frac{\tau+\sigma}{\sqrt{\tau\sigma}}>0,\qquad 
		B=\frac{\sigma-\tau}{\sqrt{\tau\sigma}}<0,
		\]
		such that \eqref{eq:q_theta} can be rewritten to yield the expressions
		\begin{equation}
			\begin{cases}
			q(\theta)=2\rho_0\sin\alpha(\theta)\cos\beta(\theta),\\
			\dot{q}(\theta)=\rho_0\bigl(A\cos\alpha(\theta)\cos\beta(\theta)
			+B\sin\alpha(\theta)\sin\beta(\theta)\bigr).
			\end{cases}
		\end{equation}
		Hence
		\[
		R(\theta)=\rho_{0}^{2}
		\begin{bmatrix}\sin\alpha \\[1ex] \cos\alpha\end{bmatrix}^{\!\!\top}\!
		M(\beta)
		\begin{bmatrix}\sin\alpha \\[1ex] \cos\alpha\end{bmatrix},
		\]
		with the symmetric matrix
		\[
		M(\beta)=
		\begin{pmatrix}
			4\cos^2\beta+B^2\sin^2\beta & AB\sin\beta\cos\beta\\[1ex]
			AB\sin\beta\cos\beta & A^2\cos^2\beta
		\end{pmatrix}.
		\]
		The eigenvalues of $M(\beta)$ satisfy
		\begin{equation}
			\begin{cases}
			\operatorname{trace} M(\beta)=(A^2+4)\cos^2\beta+B^2\sin^2\beta, \\
			\det M(\beta)=4A^2\cos^4\beta.
			\end{cases}
		\end{equation}
		Let $\lambda_{\min}(\beta)$ denote the smallest eigenvalue.  
		By the Rayleigh--Ritz principle,
		\begin{equation}
			R(\theta)\;\ge\;\rho_0^2\,\lambda_{\min}(\beta(\theta)).
			\label{eq:RR}
		\end{equation}
		
		\medskip
		\textit{1) Identification of envelope-minimising lobes:}
		Fix $\phi\in]0,1]$.
		If a lobe $I_j$ satisfies $|\cos\beta(\theta)|\ge\phi$ for all $\theta\in I_j$,
		then
		\begin{align}
		\lambda_{\min}(\beta)
		&\;\ge\;
		\frac{\det M(\beta)}{\operatorname{trace} M(\beta)}
		=
		\frac{4A^2\cos^4\beta}{(A^2+4)\cos^2\beta + B^2\sin^2\beta}\nonumber\\
		&\;\ge\;
		c(\phi)
		=
		\frac{4A^2\phi^4}{(A^2+4)\phi^2 + B^2(1-\phi^2)}\nonumber\\
		&>0.
		\end{align}
		By \eqref{eq:RR}, such a lobe satisfies
		\[
		R(\theta)\;\ge\;\rho_0^2c(\phi)\qquad (\theta\in I_j).
		\]
		Hence a global minimum cannot occur in lobes where $|\cos\beta|$ is bounded away from $0$.  
		The only candidates are lobes $I_i$ that contain points where $|\cos\beta|$ becomes arbitrarily small.  
		These are exactly the envelope-minimising lobes. Note there can exist multiple envelope-minimising lobes, but they will all lead to the same global minimum value.
		
		\medskip
		\textit{2) Monotonicity of $R$ inside an envelope-minimising lobe:}
		Fix such a lobe $I_i=[\theta_i,\theta_{i+1}]$.  
		Since the roots of $q$ and $\dot{q}$ interlace (trivial), there exists a unique interior point 
		$\theta_{c,i}\in(\theta_i,\theta_{i+1})$ such that $\dot{q}(\theta_{c,i})=0$.
		Inside the envelope-minimising lobe, $|\cos\beta|$ is small whereas $|\sin\beta|\approx 1 - O(\cos^2(\beta))$.  
		Consequently,
		\[
		q(\theta)=O(\cos\beta),\qquad q'(\theta)=O(\sin\beta),
		\]
		and thus $qq' = O(\cos^2\beta)$ is uniformly small across $I_i$, with $x'(\theta)=\frac{dx(\theta)}{d\theta}$.  
		In contrast, the dominant components of $\dot{q}$ and $\dot{q}'$ are
		\[
		\dot{q}(\theta)\approx \rho_0 B\,\sin\alpha(\theta),\qquad
		\dot{q}'(\theta)\approx \rho_0 \left[B\,\alpha' - A\,\beta'\right]\cos\alpha(\theta).
		\]  
		Thus, up to a lobe-uniform nonzero factor, the sign of $\dot{q}(\theta)\dot{q}'(\theta)$
		equals the sign of $\sin\alpha(\theta)\cos\alpha(\theta)$.
		
		Since $\alpha(\theta)$ is strictly increasing and 
		$\dot{q}(\theta_{c,i})=0$ implies $\sin\alpha(\theta_{c,i})=0$ with $\cos\alpha(\theta_{c,i})\neq 0$, 
		we have:
		\begin{equation}
			\begin{cases}
			\dot{q}\dot{q}'<0 & \quad\forall\theta\in]\theta_i,\theta_{c,i}[,\\
			\dot{q}\dot{q}'>0 & \quad\forall\theta\in]\theta_{c,i},\theta_{i+1}[.
		\end{cases}
		\end{equation}
		Note that the zero of $\dot{q}$ is simple: $\dot{q}'(\theta_{c,i})\neq 0$, since otherwise 
		$\dot{q}$ and $\dot{q}'$ would vanish simultaneously, forcing $\delta\theta\in\pi\mathbb{Z}$ 
		and $\cos\beta=0$, compatible only when $\delta\equiv 2\bmod 4$, which is excluded.
		
		Now compute
		\[
		R'(\theta)=2(qq' + \dot{q}\dot{q}').
		\]
		Because $qq'=O(\cos^{2}\beta)$ is uniformly small in $I_i$, while $\dot{q}\dot{q}'$ has strict sign on each side of $\theta_{c,i}$, it follows that
		\[
		R'(\theta)<0\ \text{on }]\theta_i,\theta_{c,i}[,\qquad
		R'(\theta)>0\ \text{on }]\theta_{c,i},\theta_{i+1}[.
		\]
		Therefore $R$ is strictly decreasing on $]\theta_i,\theta_{c,i}[$ and strictly increasing on $]\theta_{c,i},\theta_{i+1}[$.  
		Hence $\theta_{c,i}$ is the \emph{unique} minimiser of $R$ on $I_i$.  
		Any endpoint $\theta_i$ or $\theta_{i+1}$ at which $R'(\theta)$ may vanish is necessarily a local maximum, since the sign pattern of $R'$ is strictly negative immediately to the right of $\theta_i$ and strictly positive immediately to the left of $\theta_{i+1}$.
		
		\medskip
		\textit{3) Global minimality:}
		As shown in Part~1, the global minimiser must lie in an envelope‑minimising lobe.  
		Part~2 shows that within any such lobe the unique minimizer is $\theta_{c,i}$, proving the theorem.
	\end{proof}
	So from Theorem \ref{thm:globalMinimiser}, the minimal energy occurs near the slow--mode zero
	\[
	\beta=\frac{\delta}{2}\theta=\frac{\pi}{2}
	\quad\Rightarrow\quad
	\theta\approx \frac{\pi}{\delta},
	\]
	or, equivalently, near $T_{\min}\approx \pi/n$.  Furthermore, it is shown in Theorem \ref{thm:globalMinimiser} that the minimum is attained at the unique root of
	$\dot{q}(\theta)=0$  inside the interval close to $\theta_0=\pi/\delta$
	at which $|q(\theta)|$ is minimal (see Figure \ref{fig:reachable_n1_sqrt12}).  
	However, a closed form for this root is unavailable, so we introduce a systematic
	asymptotic approximation based on the slow--mode reduction.
	
	Set
	\[
	a:=\tau+\sigma,\qquad b:=\tau-\sigma,\qquad x:=\frac{b}{a}\in]0,1[,
	\]
	and so
	$
	\alpha=\frac{a}{2}\theta,\  
	\beta=\frac{b}{2}\theta = x\alpha.
	$
	Next, we propose a coordinate shift, such that the fast phase is a deviation $u$ from the $k^{th}$ node: $\alpha=k\pi+u$ with $u\in]-\pi/2,\pi/2[$ and
	$k=\mathrm{round}(1/(2x))$. We introduce
	$
	\mu=\frac{\pi}{2}-k\pi x$ and $s=\frac{\mu}{x},
	$ which is the phase lag between the fast-mode node and slow-mode node.
	Now, we propose a reduced extremum condition for $\dot{q}(\theta)=0$ which will lead to an asymptotic proxy for the true minimising phase.
	
	\begin{theorem}
		\label{thm:theta_asy}
		Let $\theta_c$ denote the unique solution of $\dot{q}(\theta)=0$ near the slow-mode node $\theta_0=\pi/b$.  
		Then the first--order asymptotic approximation of $\theta_c$ is
		\begin{equation}\label{eq:uasy}
		\theta_{\mathrm{asy}}
		=
		\frac{\pi}{b}
		+
		\frac{2}{a}\,\bigl(u_{\mathrm{asy}}-s\bigr),
		\end{equation}
		where $u_{\mathrm{asy}}$ is the unique solution of 
		\begin{equation}\label{eq:asympProxy}
			u+\tan u=s.
		\end{equation}
	\end{theorem}
	
	\begin{proof}
		Substituting $\alpha=k\pi+u$ into $\beta=x\alpha$ and using the identities
		\[
		\tan\alpha=\tan u,\qquad
		\tan^{-1}\beta=\tan(\mu-x u),
		\]
		the condition $\dot{q}(\theta)=0$ reduces exactly to
		$x\tan u=\tan(\mu-x u)$.  
		Expanding $\tan(\mu-x u)$ in $x$, yields
		$\tan(\mu-x u)=\tan\mu-x u\sec^{2}\mu+O(x^{2})$.  
		Neglecting the higher--order term, and notice that $\mu$ is small for $k$ as defined above ($\tan(\mu)\approx\mu$ and $\sec^2(\mu)\approx 1$), such that the first-order approximation
		$x\tan u\approx\tan\mu-x u\sec^{2}\mu$ simplifies to
		$u+\tan u=s$.  
		Reconstructing $\theta$ from $\alpha=\frac{a}{2}\theta=k\pi+u$ yields the
		stated expression.
	\end{proof}
	
	The next theorem provides a closed, uniform bound on the approximation
	error and quantifies the accuracy of $\theta_{\mathrm{asy}}$.
	
	\begin{theorem}
		\label{thm:error_bound}
		With $\theta_{\mathrm{asy}}$ as in Theorem~\ref{thm:theta_asy}, the
		approximation error obeys
		\begin{equation}\label{eq:errorBound}
		\bigl|\theta_c-\theta_{\mathrm{asy}}\bigr|
		\;\le\;
		\frac{\pi^{3}}{a}\,x^{2}
		=
		\pi^{3}
		\frac{(\tau-\sigma)^2}{(\tau+\sigma)^3}.
		\end{equation}
		Hence the approximation is 
		$O((\tau-\sigma)^{2}/(\tau+\sigma)^{3})$ uniformly over all coprime $(\tau, \sigma)$.
	\end{theorem}
	
	\begin{proof}
		Let $g(u) = x(\tan u + u)$. The exact condition $\dot{q}(\theta_c)=0$ is $x\tan u_c = \tan(\mu - xu_c)$. Expanding the right-hand side via Taylor's theorem yields $x\tan u_c = (\mu - xu_c) + \Delta u$, which simplifies to $g(u_c) = \mu + \Delta u$. Our approximation $u_{\text{asy}}$ satisfies $g(u_{\text{asy}}) = \mu$ by construction. By the Mean Value Theorem:
		\begin{equation}
			|u_c - u_{\text{asy}}| = \frac{|g(u_c) - g(u_{\text{asy}})|}{|g'(\xi)|} = \frac{|\Delta u|}{x(1 + \sec^2 \xi)} \le \frac{|\Delta u|}{2x},
		\end{equation}
		where $\xi \in ]u_c, u_{\text{asy}}[$ and $1+\sec^2\xi \ge 2$. Next, we bound the remainder $\Delta u$ using the Lagrange form $R_1 = \frac{1}{2}f''(\zeta)(xu)^2$ for $f(z)=\tan z$. With $|u| \le \pi/2$ and the evaluation point $|\zeta| \le |\mu| + x|u| \le \pi x$, we have:
		\begin{equation}
			|\Delta u| = |x^2 u^2 \sec^2 \zeta \tan \zeta| \le \pi^3 x^3,
		\end{equation}
		where we used $\tan(\pi x) \le 2\pi x$ and $\sec^2(\pi x) \le 2$ for small $x$. Combining these results yields $|u_c - u_{\text{asy}}| \le \frac{\pi^3 x^2}{2}$. Since $\theta = \frac{2}{a}(k\pi + u)$, the parameter error bound \eqref{eq:errorBound} is found.
		For $x \ll 1$, this $O(a^{-3})$ bound confirms the high precision of the asymptotic proxy.
	\end{proof}
	We note that the asymptotic reduction provides structural insight and a parameter-uniform approximation rather than computational complexity gains.
	However, in settings where such analytical guarantees are not required, a well-bracketed numerical method applied directly to $\dot{q}(\theta)=0$ can deliver the exact root to machine precision with comparable effort, especially when $x$ is not small or when the slow-mode localisation offers limited advantage.
	
	
	\section{Controller design}\label{sec:controllerDesing}
		
	From a control theoretical point of view, we can discern $(q,\dot q)$ as a host subsystem and $(z,\dot z)$ as an auxiliary/controller, connected through a power‑preserving port with storage function $H$ \cite{Juchem2024}. The inscribed radius bound $r_{res}$ furnishes an \emph{internal performance} certificate consistent with passivity and dissipativity theory \cite{vanDerSchaftL2,HaddadChellaboina,Willems2007,HillMoylan1976}.
	
	As mentioned in Section \ref{sec:subsystem}, consider the closed-loop class induced by shaping the constant gyroscopic interconnection strength $n\in\mathbb{R}$ (and, optionally, kinetic/potential scalings that preserve the quadratic Hamiltonian form), subject to the impulse disturbance class $\mathcal{D}$ in \eqref{eq:IC}.	With the above, we can obtain for each choice of $n$ and each disturbance in $\mathcal{D}$ the largest origin-centered circle with radius $r_{\mathrm{res}}(n;\dot{q}_0)$ inscribed in the \emph{projected} subsystem trajectory on $(q,\dot q)$ (or equivalently on $(z,\dot z)$ by symmetry).
	By linearity, $r_{\mathrm{res}}(n;\dot{q}_0)=|\dot{q}_0|\,r_{\mathrm{res}}(n;1).$
	Equivalently, $\min_{t\ge 0} H_q(t) = \frac{1}{2}\, r_{\mathrm{res}}(n;\dot{q}_0)^2.$
	
	Both the largest inscribed circle radius $r_{\mathrm{res}}$ and the estimated time to reach said minimum $T_{min}$ (for the first time) can be used to obtain an optimal controller design. Recall that  $T_{min}\approx\frac{\pi}{|n|}$, and in a resonant case $T_{min}=\sqrt{\tau\sigma}\theta_{c}$. Two design problems can be discerned:
	
	\begin{itemize}
	\item\textbf{Absorption with speed constraint:}
	\begin{equation}\label{eq:design-absorb}
		\min_{n\in\mathcal{N}}\;\; r_{\mathrm{res}}(n;D)
		\qquad\text{s.t.}\quad
		T_{\min}(n)\le T_{\max},
	\end{equation}
	and the optional exclusion $\Omega_1\!:\!\Omega_2\notin \mathcal{R}_{\mathrm{low}}$ avoids low‑order locking sets that inflate $r_{\mathrm{res}}$.
	
	\item\textbf{Containment with responsiveness: }
	\begin{equation}\label{eq:design-contain}
		\max_{n\in\mathcal{N}}\;\; r_{\mathrm{res}}(n;D)
		\qquad\text{s.t.}\quad
		T_{\min}(n)\le T_{\max}.
	\end{equation}	
	\end{itemize}
	
	The set $\{(r_{\mathrm{res}}(n;1),\,T_{\min}(n))\}_{n\in\mathcal{N}}$ defines a Pareto frontier. Low‑order locking points (large $n$) give fast but conservative exchange (large $r_{\mathrm{res}}$), whereas near‑irrational points (small $n$) give deep exchange (small $r_{\mathrm{res}}$) but slow attraction. This is clearly demonstrated in Figure \ref{fig:pareto}, where $r_{res}$ is plotted in function $\log{1+T_{min}}$ for different resonant pairs ($\tau, \sigma$). The red line shows a clear trade-off between both performance metrics.
	\begin{figure}[h!]
		\centering
		\includegraphics[width=0.75\columnwidth]{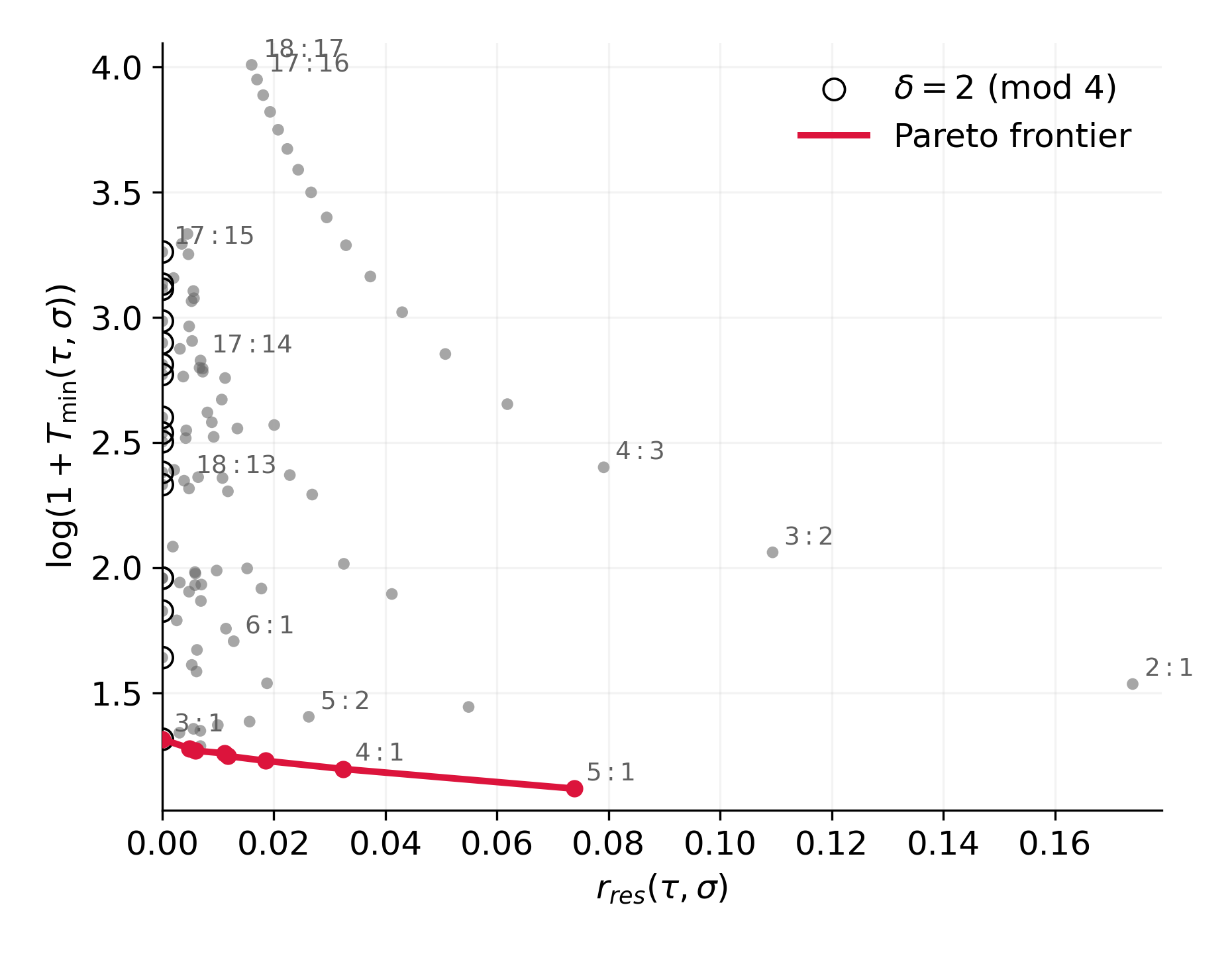}
		\caption{The Pareto frontier of speed vs. minimum energy. Schematic plot of $\big(r_{\mathrm{res}}(n(\tau, \sigma);1),\,T_{\min}(n(\tau,\sigma))\big)$ as $(\tau,\sigma)$ varies. Note the low‑order resonances (fast/large) and near‑irrational ratios (slow/small).}
		\label{fig:pareto}
	\end{figure}

	\section{Examples}
	\label{sec:examples}
	We consider the gyroscopically coupled, conservative 2-DOF system \eqref{eq:model} with an impulse disturbance \eqref{eq:IC} with $\dot q_0=1$. We present two examples where $n$ needs to be designed. Following prior absorber design practice, we require a pronounced beat:
	\begin{equation}\label{eq:beatqual}
	\frac{\tau+\sigma}{\tau-\sigma}\geq 10,
	\end{equation}
	so that one slow envelope period contains at least ten fast cycles.
	
	\subsection{Case A - Absorption}
	We target \(\delta\equiv2\ (\mathrm{mod}\ 4)\) to enable complete cancellation (zero inscribed radius) and preferably as fast as possible, as in \eqref{eq:design-absorb}.
	With $\delta=2$, \eqref{eq:beatqual} gives $\tau \geq 11$ and, thus, $\sigma\geq9$, and the pair (11, 9) is coprime. This gives $n = 0.201$ and $T_{min} \approx 15.6s$.
	With $\delta=6$, \eqref{eq:beatqual} gives $\tau \geq 33$, so the smallest coprime that agrees to these conditions is (41, 35), but $T_{min}\approx 19.8s$, so this is worst and will only deteriorate if we increase $\delta$ further (see Figure \ref{fig:pareto}). 
	So for $(\tau, \sigma)=(11, 9)$, the Lissajous trajectory in \((q,\dot q)\) is closed and intersects the origin, so \(r_{\min}=0\). Furthermore is the beat pronounced and is the time to achieve zero energy minimal. Figure \ref{fig:absorption_phase} shows the phase portrait with the closed, sparse Lissajous curve, the analytic convex boundary \eqref{eq:boundary}, and the largest origin-centered inscribed circle (degenerate).
	Figure \ref{fig:absorption_energy} shows \(H_q(t)\) with the predicted \(T_{\min}\) marker.
	
	\subsection{Case B - Containment}
	Here, we want to achieve design criterium \eqref{eq:design-contain}. We require resonance for a sparse trajectory but avoid \(\delta\equiv2\ (\mathrm{mod}\ 4)\). From Figure \ref{fig:pareto} it is clear that $\delta=1$ leads to large $r_{res}$.
	Then from \eqref{eq:beatqual} if follows that $\tau \geq 5.5$. As $\tau \in\mathbb{N}$ the coprime pair $(\tau, \sigma)=(6,5)$ is chosen.	Then $n= 0.183$, and $T_{\min}\approx 17.2s$.
	Here, the Lissajous trajectory is closed and remains away from the origin. To find the resonant inscribed radius, we find from \eqref{eq:asympProxy} that $u_{asy}=-0.17rad$, which gives $\theta_{asy}=3.30rad$. This is validated as $\dot{q}(\theta_{asy})=2.43\times 10^{-4}$, and \eqref{eq:errorBound} leads to $|\theta_c - \theta_{asy}|=1.16\times 10^{-4} \le 2.33\times 10^{-2}$. This yields $r_{\min}=5.075\times 10^{-2}$, i.e., a certified lower bound on $H_q$. 
	Figure~\ref{fig:containment_phase} shows the phase portrait with the trajectory, envelope, and the non-zero inscribed circle.
	Figure \ref{fig:containment_energy} shows \(H_q(t)\) and the \(T_{\min}\) marker.
	
	\begin{figure}[!t]
		\centering
		\begin{subfigure}{0.48\linewidth}
			\includegraphics[width=\linewidth]{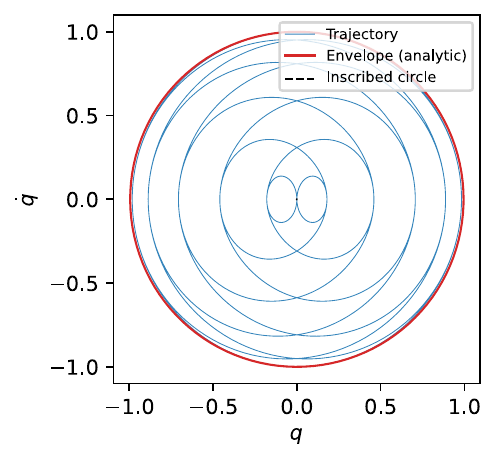}
			\caption{}
			\label{fig:absorption_phase}
		\end{subfigure}
		\hfill
		\begin{subfigure}{0.48\linewidth}
			\includegraphics[width=\linewidth]{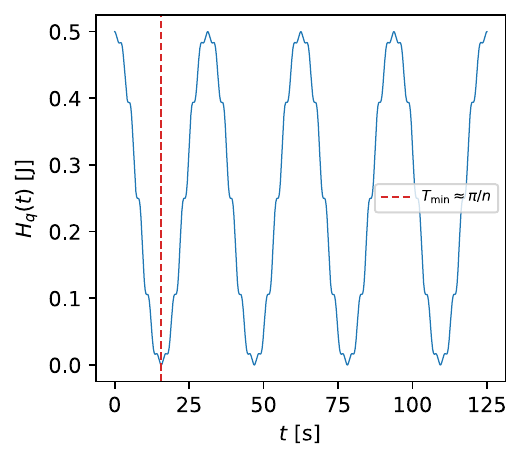}
			\caption{}
			\label{fig:absorption_energy}
		\end{subfigure}
		\caption{Absorption case $(11,9)$: (a) phase portrait with envelope and inscribed circle, and (b) corresponding subsystem energy. The minimum occurs at $T_{\min}\approx \pi/n$.}
		\label{fig:absorption}
	\end{figure}
	
	\begin{figure}[!t]
		\centering
		\begin{subfigure}{0.48\linewidth}
			\includegraphics[width=\linewidth]{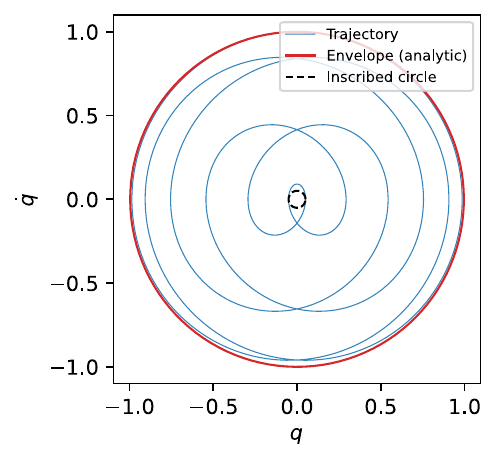}
			\caption{}
			\label{fig:containment_phase}
		\end{subfigure}
		\hfill
		\begin{subfigure}{0.48\linewidth}
			\includegraphics[width=\linewidth]{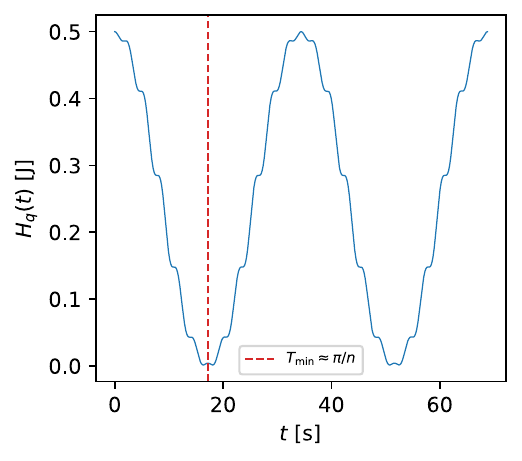}
			\caption{}
			\label{fig:containment_energy}
		\end{subfigure}
		\vspace{-1mm}
		\caption{Containment case $(6,5)$: (a) phase portrait with envelope and inscribed circle, showing a non-zero inscribed radius, and (b) the corresponding subsystem energy.}
		\label{fig:containment}
		\vspace{-2mm}
	\end{figure}
	
%
%
%
	
	\section{Conclusions}\label{sec:conclusions}
	We presented a geometric and computational analysis of energy exchange in 2-DOF gyroscopic systems, deriving exact envelopes, explicit resonance conditions, and introducing the resonant inscribed radius as a metric to quantify the minimal energy that remains confined to a subsystem for any resonant pair $(\tau, \sigma)$, which determines a unique coupling parameter $n$. Degeneracy of this metric occurs precisely when $\tau - \sigma \equiv 2 \pmod{4}$. In all other cases, we provided a numerical procedure to compute the resonant inscribed radius together with a uniform error bound. Finally, we connected these insights to a controller design framework for targeted energy absorption and containment within a subsystem, and illustrated the methodology through representative examples.
	
	\bibliographystyle{IEEEtran}
	\bibliography{biblio}
\end{document}